\documentclass[11pt]{article}

\usepackage[hyphens]{url}
\usepackage[pagebackref,colorlinks]{hyperref}
\usepackage[hyphenbreaks]{breakurl}
\usepackage{amsmath} 
\usepackage{amsthm} 
\usepackage{amssymb}	
\usepackage{graphicx} 
\usepackage{multicol} 
\usepackage{multirow}
\usepackage{color}
\usepackage[dvips,letterpaper,margin=1in,bottom=1in]{geometry}
\usepackage[capitalize,noabbrev]{cleveref}

\usepackage{bm}
\usepackage{bbm}

\usepackage{ytableau}
\usepackage{longtable}

\usepackage[utf8]{inputenc}
\usepackage[english]{babel}

\usepackage[T1]{fontenc}
\AtBeginDocument{%
  \DeclareFontShape{T1}{cmr}{m}{scit}{<->ssub*cmr/m/sc}{}%
}

\usepackage{diagbox}
\usepackage{mathtools}
\allowdisplaybreaks

\newtheorem{theorem}{Theorem}[section]

\newtheorem{lemma}[theorem]{Lemma}

\newtheorem{fact}[theorem]{Fact}

\newtheorem{problem}{Problem}

\DeclarePairedDelimiter\rbra{\lparen}{\rparen}
\DeclarePairedDelimiter\sbra{\lbrack}{\rbrack}
\DeclarePairedDelimiter\cbra{\{}{\}}
\DeclarePairedDelimiter\abs{\lvert}{\rvert}
\DeclarePairedDelimiter\Abs{\lVert}{\rVert}
\DeclarePairedDelimiter\ceil{\lceil}{\rceil}
\DeclarePairedDelimiter\floor{\lfloor}{\rfloor}

\DeclareMathOperator*{\E}{\mathbb{E}}
\DeclareMathOperator*{\Var}{\mathbb{V}\!\!\;ar}

\newcommand{\supp} {\operatorname{supp}}

\usepackage{latexsym}
\usepackage{CJK}

\usepackage{enumerate}
\usepackage{enumitem}

\usepackage{algorithm}
\usepackage{algpseudocode}

\usepackage{stmaryrd}
\usepackage{tabularx}
\usepackage{booktabs}
\usepackage{adjustbox}

\newcommand{\footremember}[2]{%
    \footnote{#2}
    \newcounter{#1}
    \setcounter{#1}{\value{footnote}}%
}

\usepackage{tikz}

\usetikzlibrary{quantikz2}

\begin{document}

\title{Optimal Support Size Testing}
\author{Qisheng Wang \footremember{1}{Qisheng Wang is with the School of Computer Science, Shanghai Jiao Tong University (e-mail: \url{QishengWang1994@gmail.com}).}}
\date{}

\maketitle

\begin{abstract}

    We present an optimal tester for support size that determines whether an unknown probability distribution is supported on $n$ elements or $\varepsilon$-far away (in total variation distance) with sample complexity $\Theta(\frac{n}{\varepsilon\log(n)})$,
    improving the prior best $O(\frac{n}{\varepsilon\log(n)} \log(\min\{\frac{1}{\varepsilon}, n\}))$ due to \hyperlink{cite.FPJH25}{Ferreira Pinto Jr.\ and Harms (STOC 2025)}.
    As a generalization, we provide a \textit{pretty} good lower bound on support size. 
\end{abstract}



\section{Introduction}

The support size of an unknown discrete probability distribution is an important measure of species diversity in ecology \cite{FCW43}. 
The estimation of support size has been subsequently investigated in statistics from a theoretical perspective \cite{Goo49,Goo53} (cf.\ \cite{BF93}). 
Since ``I don't believe it is usually possible to estimate the number of unseen species [i.e., support size]'' (quoting I.\ J.\ Good in \cite{BF93}), the best we can hope is to consider weaker tasks with promise rather than estimation in all cases. 
Among them, support size testing is a basic statistical decision problem in property testing (cf.\ \cite[Section 11.4]{Gol17}). 

\begin{problem}[Support size testing]
    Given an unknown probability distribution $p$, the task is to determine whether $p$ is supported on at most $n$ elements or $\varepsilon$-far in total variation distance from any distribution supported on at most $n$ elements. 
    Formally, determine whether 
    \[
    \textup{Yes: } \abs{\supp\rbra{p}} \leq n \qquad \text{or} \qquad \textup{No: } \Delta_n\rbra{p} \geq \varepsilon,
    \]
    with the promise that it is in either case. 
    Here, $\supp\rbra{p} = \cbra{i \colon p_i > 0}$ denotes the set of elements with nonzero probability, and $\Delta_n\rbra{p} = \inf_{q \colon \abs{\supp\rbra{q}} \leq n} d_{\textup{TV}}\rbra{p, q}$ denotes the total variation distance of $p$ to distributions supported on at most $n$ elements, where $d_{\textup{TV}}\rbra{p, q} = \frac{1}{2} \sum_i \abs{p_i - q_i}$. 
\end{problem}

This problem has been investigated in a series of works \cite{GR22,AF24,AFL24,KLR25,FPJH25}.
In particular, the prior best approach is with sample complexity $O\rbra{\frac{n}{\varepsilon\log\rbra{n}}\log\rbra{\min\cbra{\frac{1}{\varepsilon}, n}}}$ due to \cite{FPJH25}, improving the prior approaches via learning histograms \cite{VV11,VV17,HJW18,HO19} with sample complexity $O\rbra{\frac{n}{\varepsilon^2\log\rbra{n}}}$.
However, the prior best lower bound is $\Omega\rbra{\frac{n}{\varepsilon \log\rbra{n}}}$ implied by \cite[Theorem 2]{VV11} and \cite[Proposition 1]{WY19}. 
There still remains a gap of $\log\rbra{\min\cbra{\frac{1}{\varepsilon}, n}}$ between the upper and lower bounds. 
In this paper, we present an optimal support size tester, thereby closing this gap. 

\begin{theorem} [Optimal support size testing] \label{thm:support-size-testing}
    There is an algorithm for support size testing with optimal sample complexity $\Theta\rbra{\frac{n}{\varepsilon\log\rbra{n}}}$.
\end{theorem}

\subsection{Generalization: Pretty good lower bound on support size} \label{sec:pretty}

As a generalization of this result, we can further give a \textit{pretty} good lower bound on the support size, without any assumptions on the distribution. 

\begin{theorem}[Pretty good lower bound on support size] \label{thm:pretty}
    There is an algorithm that uses $\Theta\rbra{\frac{n}{\varepsilon\log\rbra{n}}}$ samples drawn from an unknown distribution $p$ and outputs a number $\widehat{S} \in \cbra{1, 2, \dots, n}$ such that with probability at least $2/3$, 
    \begin{equation} \label{eq:def-hat-S}
    \min\cbra{\operatorname{eff}_{\varepsilon}\rbra{p}, n} \leq \widehat{S} \leq \abs{\supp\rbra{p}},
    \end{equation}
    where $\operatorname{eff}_{\varepsilon}\rbra{p}$ is the $\varepsilon$-effective support size of $p$,\footnote{The effective support size was recently involved in property testing \cite{CDS18,BCG19}, and its estimation (in a query model) was studied in \cite{Gol19,NT23,Gol25}. } the smallest number $k \in \mathbb{N}$ such that there is a distribution $q$ with $\abs{\supp\rbra{q}} = k$ and $d_{\textup{TV}}\rbra{p, q} \leq \varepsilon$. 
\end{theorem}

If an estimate $\widehat{S}$ satisfies \cref{eq:def-hat-S}, then we call it a \textit{pretty good lower bound on support size}, which is a stronger definition than the \textit{good lower bound on support size} $\min\cbra{\operatorname{eff}_{\varepsilon}\rbra{p}, n} \leq \widehat{S} \leq \rbra{1+\varepsilon}\abs{\supp\rbra{p}}$ defined in \cite[Corollary 1.6]{FPJH25}, where we remove the factor $\rbra{1+\varepsilon}$. 
\cref{thm:pretty} can also be interpreted as: using $m$ samples, one can obtain a lower bound on support size at least $\min\cbra{\operatorname{eff}_\varepsilon\rbra{p}, \Theta\rbra{\varepsilon m \log \rbra{e + \varepsilon m}}}$, which improves the lower bound $\min\cbra{\operatorname{eff}_\varepsilon\rbra{p}, \Theta\rbra{\frac{\varepsilon}{\log\rbra{1/\varepsilon}}m\log\rbra{m}}}$ in \cite{FPJH25}. 

The algorithm in \cref{thm:pretty} is also optimal since a pretty good lower bound on support size can solve the support size testing problem: $\abs{\supp\rbra{p}} \leq \floor{n/2}$ vs $\Delta_{\floor{n/2}}\rbra{p} \geq 2\varepsilon$,\footnote{If $\abs{\supp\rbra{p}} \leq \floor{n/2}$, then $\widehat{S} \leq \floor{n/2}$. If $\Delta_{\floor{n/2}}\rbra{p} \geq 2\varepsilon$, then $\operatorname{eff}_\varepsilon\rbra{p} > \floor{n/2}$, which gives $\widehat{S} \geq \floor{n/2}+1$.} and thus the lower bound $\Omega\rbra{\frac{n}{\varepsilon\log\rbra{n}}}$ for support size testing \cite{VV11,WY19} applies. 

\subsection{Implications in distribution-free sample-based testing}

In \cite{FPJH25}, they provided an algorithm for testing the class $\mathcal{H}_n$ of Boolean functions $f \colon \mathbb{N} \to \cbra{0, 1}$ with $\abs{f^{-1}\rbra{1}} \leq n$ with sample complexity $O\rbra{\frac{n\log\rbra{1/\varepsilon}}{\varepsilon\log\rbra{n}}}$ in the distribution-free sample-based testing model \cite{GR16,Gol17,RR20,BFH21,RR22}. 
Testing Boolean function classes can aid the proper PAC (probably approximately correct) learning in model selection when testing is faster than learning. 
As an implication of our results, combining \cite[Theorem 1.8]{FPJH25}, we settle the sample complexity of testing $\mathcal{H}_n$.

\begin{theorem}[Optimal distribution-free sample-based testing for $\mathcal{H}_n$] \label{thm:Hn}
    There is an algorithm that uses $\Theta\rbra{\frac{n}{\varepsilon\log\rbra{n}}}$ samples of the form $\rbra{x, f\rbra{x}}$ with $x \sim p$ for an unknown function $f \colon \mathbb{N} \to \cbra{0, 1}$ and an unknown probability distribution $p$, and determines whether $f \in \mathcal{H}_n$ or $f$ is $\varepsilon$-far from $\mathcal{H}_n$ (i.e., $\inf_{g \in \mathcal{H}_n} \Pr_{x \sim p}\sbra{f\rbra{x} \neq g\rbra{x}} \geq \varepsilon$) with success probability at least $2/3$.
\end{theorem}
\begin{proof}
    \cite[Theorem 1.8]{FPJH25} gives $m_{\mathcal{H}_n}\rbra{n, \varepsilon} \leq m_{\textup{supp}}\rbra{n, \varepsilon} + O\rbra{1/\varepsilon}$ where $m_{\mathcal{H}_n}\rbra{n, \varepsilon}$ is the sample complexity of testing $\mathcal{H}_n$ and $m_{\textup{supp}}\rbra{n, \varepsilon}$ is the sample complexity of support size testing.
\end{proof}

The upper bound in \cref{thm:Hn} matches the lower bound $\Omega\rbra{\frac{\mathsf{VC}}{\varepsilon\log\rbra{\mathsf{VC}}}}$ due to \cite{BFH21} (cf.\ \cite[Section 1.3]{FPJH25}), where the VC (Vapnik--Chervonenkis) dimension of $\mathcal{H}_n$ is $\mathsf{VC}\rbra{\mathcal{H}_n} = n$. 

\subsection{Techniques}

We now explain the intuition of our tester.
Given the inspiration from prior approaches, we use a new strategy to solve an even harder problem: \textit{pretty good lower bound on support size}.
During directly facing this harder problem, it suggests new forms of testers and better parameter choices. 
Since our tester itself (as well as its correctness proof) is very \textit{simple}, for the readers who are interested in the specific construction, skip to the last paragraph of this subsection and \cref{sec:tester}. 

Like prior works \cite{JVHW15,WY16} on property testing, we employ the Poissonized sampling model: Let $M \sim \operatorname{Poi}\rbra{m}$, draw $M$ independent samples $X_1, X_2, \dots, X_M$ from $p$, and the algorithm takes the empirical counts $N_i = \sum_{j} \mathbf{1}_{\cbra{X_j=i}}$ as input. 
The sample complexity in the Poissonized sampling model is essentially the same as that in the standard (multinomial) sampling model up to a constant factor; namely, the algorithm takes independent samples $X_1, X_2, \dots, X_m \sim p$ as input.

\paragraph{The approach of \cite{FPJH25}.}
In \cite{FPJH25}, they first provided a tester for support size and then used it as a subroutine to give a good lower bound on support size. 
Their tester is of the form 
\[
\widehat{S}^{\textup{FH}} \coloneqq \widehat{S}_{a,b,k,m} = \sum_i f_{a,b,k,m}\rbra{N_i},
\]
with $a \asymp \frac{\varepsilon}{n}\log\rbra{\frac{1}{\varepsilon}}$, $b \asymp \frac{\varepsilon \log^2\rbra{n}}{n\log\rbra{1/\varepsilon}}$, $k \asymp \log\rbra{n}$, and $m \asymp \frac{n\log\rbra{1/\varepsilon}}{\varepsilon\log\rbra{n}}$, 
where
\begin{align}
    P_{a,b,k}\rbra{x} & = -\rho T_k\rbra*{\frac{b+a-2x}{b-a}} = \sum_{\ell=0}^k c_\ell x^\ell, \qquad \rho = \rbra*{T_k\rbra*{\frac{b+a}{b-a}}}^{-1} \in (0, 1], \label{eq:def-Pabk} \\
    f_{a,b,k,m}\rbra{\ell} & = \begin{cases}
        1 + c_\ell \ell! / m^{\ell}, & 1 \leq \ell \leq k, \\
        1, & \ell > k, \\
        0, & \ell = 0,
    \end{cases} \qquad T_k\rbra{x} = \begin{cases}
    1, & k = 0, \\
    x, & k = 1, \\
    2x T_{k-1}\rbra{x} - T_{k-2}\rbra{x}, & k \geq 2. \label{eq:def-fabkm}
\end{cases}
\end{align}
Their tester returns ``Yes'' if $\widehat{S}^{\textup{FH}} < \rbra{1+\frac{\varepsilon}{2}}n$ and ``No'' otherwise.
Here, $T_k\rbra{x}$ is the $k$-th Chebyshev polynomial of the first kind. 
It is noted that $P_{a,b,k}\rbra{\cdot}$ and $f_{a,b,k,m}\rbra{\cdot}$ were also used in \cite{WY19} for support size estimation.

\paragraph{New strategy for testing.}
Rather than directly focusing on support size testing, we aim to deal with giving a pretty good lower bound on support size as well. 
Namely, our goal is to directly find a pretty good lower bound $\widehat{S}$ on support size such that $\min\cbra{\operatorname{eff}_{\varepsilon}\rbra{p}, n} \leq \widehat{S} \leq \abs{\supp\rbra{p}}$. 
If this can be done, it already gives a support size tester, as explained in \cref{sec:pretty}. 
So, in the following, we skip directly to solving this harder task and we will explain why the harder task will give hints to an optimal support size tester. 

\paragraph{Step 1: Removing the multiplicative factor $\rbra{1+\varepsilon}$.}
The reason that the good lower bound on support size obtained in \cite{FPJH25} has an upper bound of $\rbra{1+\varepsilon}\abs{\supp\rbra{p}}$ with a multiplicative factor $\rbra{1+\varepsilon}$ is that if $\abs{\supp\rbra{p}} \leq n$, the best possible upper bound on $\E\sbra{\widehat{S}^{\textup{FH}}}$ we can have is already $\E\sbra{\widehat{S}^{\textup{FH}}} \leq \rbra{1+\delta} \abs{\supp\rbra{p}}$ where $\delta = \varepsilon/4$ (see \cite[Lemma 3.7]{FPJH25}). 
On the other hand, $\abs{\widehat{S}^{\textup{FH}} - \E\sbra{\widehat{S}^{\textup{FH}}}} \leq \delta n$ with high probability (see \cite[Lemma 3.17]{FPJH25}). 
The two inequalities inspire us to consider the case where there is a tester $\widehat{S}$ such that with high probability,
\[
\E\sbra{\widehat{S}} \leq \rbra{1 + \delta} \abs{\supp\rbra{p}}, \qquad \abs*{\widehat{S} - \E\sbra{\widehat{S}}} \leq \delta n,
\]
where $\delta = \varepsilon/4$ reduces to the case of \cite{FPJH25}.
The two conditions together yield 
\[
\widehat{S} \leq \rbra{1+\delta} \abs{\supp\rbra{p}} + \delta n,
\]
which is 
\[
\frac{\widehat{S} - \delta n}{1 + \delta} \leq \abs{\supp\rbra{p}}. 
\]
This suggests we use a tester of the form
\[
\widehat{S}_{a,b,k,m,\delta}' = \frac{\widehat{S}_{a,b,k,m} - \delta n}{1 + \delta},
\]
which guarantees $\widehat{S}_{a,b,k,m,\delta}' \leq \abs{\supp\rbra{p}}$ with high probability.
Since our goal is the sample complexity $m \asymp \frac{n}{\varepsilon\log\rbra{n}}$, it remains to determine the rest parameters. 

\paragraph{Step 2: The choice of parameter $a$ leads to multiple sub-testers.}
We now consider how to choose $a$ so that $\widehat{S}'_{a,b,k,m,\delta} \geq \operatorname{eff}_\varepsilon\rbra{p}$ when $\operatorname{eff}_\varepsilon\rbra{p} \leq n$. 
In other words, for every $1 \leq r \leq n$, if $\operatorname{eff}_\varepsilon\rbra{p} > r$, then $\widehat{S}'_{a,b,k,m,\delta} > r$. 
Note that under some smooth conditions, one has $\E_{X \sim \operatorname{Poi}\rbra{mx}} \sbra{f_{a,b,k,m}\rbra{X}} \geq \rbra{1-\delta_a} \min\cbra{x/a, 1}$ (see \cref{lemma:poly}),
where $\delta_a > 0$ depends on $a$. 
When $\operatorname{eff}_\varepsilon\rbra{p} > r$, this further leads to 
\[
\widehat{S}'_{a,b,k,m,\delta_a} - r \geq \frac{\rbra{1-\delta_a}\varepsilon/p_{\rbra{r}} - 2\delta_a \rbra{n+r}}{1+\delta_a},
\]
where $p_{\rbra{r}}$ denotes the $r$-th largest probability weight of $p$. 
So, in order to ensure $\widehat{S}'_{a,b,k,m,\delta_a} > r$, it suffices to ensure $\rbra{1-\delta_a}\varepsilon/p_{\rbra{r}} > 2\delta_a \rbra{n+r}$, which, for small $\delta_a \leq 1/4$, is 
\begin{equation} \label{eq:key-cond}
\delta_a n \lesssim \frac{\varepsilon}{p_{\rbra{r}}}. 
\end{equation}

However, $p_{\rbra{r}}$ is not known yet. 
To overcome this issue, we consider a list of candidates of $a$, namely $a_0, a_1, \dots, a_J$ with $a_j = 2^j a_0$. 
For large enough $J$, there must be a $0 \leq j < J$ such that $a_j \leq p_{\rbra{r}} < a_{j+1} = 2a_j$, in which case $p_{\rbra{r}} \asymp a_j$ and one can take $\delta_{a_j} \asymp \frac{\varepsilon}{a_jn}$ to meet \cref{eq:key-cond}. 
To complete this, since smaller $a$ gives larger $\delta_a$ while we require in our analysis that $\delta_a \leq 1/4$, the smallest possible $a$ (i.e., $a_0$) we can choose is therefore $a_0 \asymp \frac{\varepsilon}{\delta_{a_0} n}$, which is $a_0 \asymp \frac{\varepsilon}{n}$ when $\delta_a \asymp 1$. 
As a result, we can take $a_j \asymp \frac{2^j \varepsilon}{n}$ and $\delta_j \asymp 2^{-j}$ and these suggest a tester of the form 
\[
\widehat{S}' = \max_{0 \leq j \leq J} \widehat{S}'_{a_j, b_j, k_j, m, \delta_j}.
\]
For the rest of the parameters, $k_j \asymp \sqrt{\frac{b_j}{a_j}}\log\rbra{\frac{1}{\delta_j}} \lesssim \log\rbra{n}$ is chosen to ensure that $\rho = \rbra{T_{k_j}\rbra{\frac{b_j+a_j}{b_j-a_j}}}^{-1} < \delta_{j}$, and $b_j \asymp \frac{\log\rbra{n}}{m} \asymp \frac{\varepsilon\log^2\rbra{n}}{n}$ is chosen to ensure that every element with top-$n$ probability is seen with high probability. 
As a result, $J \asymp \log\rbra{\frac{b}{a_0}} \asymp \log\rbra{\log\rbra{n}}$. 

\paragraph{Step 3: Maximum likelihood estimation strikes back.}
However, the tester $\widehat{S}'$ may fail in some extreme cases.
For example, when $p_1 = p_2 = \dots = p_{n-1} = \frac{1-\varepsilon}{n-1}$ and $p_{n} = \varepsilon$ with $\varepsilon = \frac{1}{4n}$, $\Delta_{n-1}\rbra{p} = \varepsilon$ and $\operatorname{eff}_\varepsilon\rbra{p} = \abs{\supp\rbra{p}} = n$, which forces the pretty good lower bound on the support size of $p$ to be exactly $\widehat{S} = n$. 
On the other hand, with high probability, $\widehat{S}_{a_j,b_j,k_j,m} = n$, which gives $\widehat{S}'_{a_j, b_j, k_j, m, \delta_j} \leq n-1$ as long as $\delta_j n \gtrsim 1$ and therefore $\widehat{S}' \leq n-1$. 

To overcome this issue, we employ the simple maximum likelihood estimator 
\begin{equation} \label{eq:def-MLE}
\widehat{S}^{\textup{MLE}} = \sum_i \mathbf{1}_{\cbra{N_i > 0}},
\end{equation}
and finally our tester is now completed:
\[
\widehat{S}^{\textup{our}} = \max\cbra*{ \widehat{S}^{\textup{MLE}}, \widehat{S}' }.
\]
It only remains to verify that the tester $\widehat{S}^{\textup{our}}$ can give a pretty good lower bound on support size and thus a support size tester. 

\paragraph{Our tester: A formal description.}

\[
\widehat{S}^{\textup{our}} = \max\cbra*{ \widehat{S}^{\textup{MLE}}, \max_{0 \leq j \leq J} \frac{\widehat{S}^{\textup{Cheb}}_j - \delta_j n}{1+\delta_j} }, \qquad \widehat{S}^{\textup{Cheb}}_j = \widehat{S}_{a_j, b_j, k_j, m},
\]
where $m \asymp \frac{n}{\varepsilon\log\rbra{n}}$ and $a_j \asymp \frac{2^j \varepsilon}{n}$, $b_j \asymp \frac{\varepsilon \log^2\rbra{n}}{n}$, $k_j \lesssim \log\rbra{n}$, $\delta_j \asymp 2^{-j}$, and $J \asymp \log\rbra{\log\rbra{n}}$.
Note that the choice of $b_j$ is independent of $j$. 
With the choice of these parameters, 
\begin{enumerate}
    \item $\widehat{S}^{\textup{our}}$ gives a pretty good lower bound on support size. 
    \item The support size tester is immediate: return ``Yes'' if $\widehat{S}^{\textup{our}} \leq n$ and ``No'' otherwise. 
\end{enumerate}

\subsection{Independent and Simultaneous Work}

During the preparation of this paper, the author became aware of the independent work \cite{BGY26} with an optimal support size tester over a finite alphabet.
Although both approaches employ polynomial estimators, their algorithmic construction is substantially different. 
Specifically, they first amplify the samples by taking unions, construct an initial set of observed elements, and then reduce the problem to testing bounded residual sets, involving the Fej\'er polynomial. 
In contrast, our approach directly deals with the empirical counts using sub-testers involving affined Chebyshev polynomials. 
Furthermore, our approach yields a pretty good lower bound on support size, providing a numerical guarantee beyond support size testing. 

\section{Preliminaries}

\begin{fact}[{\cite[Exercise 4.7]{MU05}}]
    For $M \sim \operatorname{Poi}\rbra{m}$, then $\Pr\sbra{M > 4m} \leq e^{-m}$. 
\end{fact}

\begin{fact} \label{lemma:mle}
    Let independent $N_i \sim \operatorname{Poi}\rbra{mp_i}$. 
    Define $\widehat{S}^{\textup{MLE}}$ as in \cref{eq:def-MLE}. 
    Then, 
    \begin{equation} \label{eq:basic}
    \Var\sbra*{\widehat{S}^{\textup{MLE}}} \leq \E\sbra*{\widehat{S}^{\textup{MLE}}} = \sum_{i} \rbra*{1 - e^{-mp_i}} \leq \abs{\supp\rbra{p}}. 
    \end{equation}
    If $\Delta_r\rbra{p} \geq \varepsilon$, then 
    \begin{equation} \label{eq:prob-mle}
    \Pr\sbra*{ \widehat{S}^{\textup{MLE}} \leq r } \leq \frac{r+1}{m\varepsilon}.
    \end{equation}
\end{fact}

\begin{lemma}[Adapted from {\cite[Proposition 2.12, 3.2, 3.4, and 3.14]{FPJH25}}]\label{lemma:poly}
    Let $0 < a < b \leq 1$, $m > 0$ and integer $k \geq 1$. 
    Let $P_{a,b,k}\rbra{x}$ and $f_{a,b,k,m}\rbra{\ell}$ be defined in \cref{eq:def-Pabk,eq:def-fabkm}, respectively. 
    Then, 
    \begin{equation} \label{eq:def-Q}
    Q\rbra{x} \coloneqq \E_{X \sim \operatorname{Poi}\rbra{mx}} \sbra*{f_{a,b,k,m}\rbra{X}} = 1 + e^{-mx} P_{a,b,k}\rbra{x}, \qquad x \geq 0.
    \end{equation}
    If $m\sqrt{b\rbra{b-a}} \geq k$, then 
    \begin{equation} \label{eq:outcond2}
        \rbra{1-\rho} \min\cbra{x/a, 1} \leq Q\rbra{x} \leq 1 + \rho.
    \end{equation}
    If $a \leq b/2$, then for $1 \leq \ell \leq k$, 
    \begin{equation} \label{eq:outcond3}
        \abs{c_\ell} \ell! \leq 15^k \rbra*{\frac{k}{b}}^{\ell}.
    \end{equation}
\end{lemma}
\begin{proof}
    See \cref{app:proof-lemma-poly}.
\end{proof}

\section{The Tester} \label{sec:tester}

Set constants $C_0 = 2^{26}$ and $L_0 = 2^{10}$.
The choice of constants is for simplicity, and we did not try to optimize them. 
Let 
\[
L = \log\rbra{en}, \qquad m = \ceil*{\frac{C_0n}{\varepsilon L}}. 
\]
Let $M \sim \operatorname{Poi}\rbra{m}$ and let $X_1, X_2, \dots, X_M \sim p$ be $M$ independent samples from $p$. 
Let
\[
N_i = \sum_{j=1}^M \mathbf{1}_{\cbra{X_j = i}}. 
\]
Note that $N_i \sim \operatorname{Poi}\rbra{m p_i}$ are independent. 
Our tester considers two cases. 
\begin{enumerate}
    \item $L < L_0$. 
    Let $\widetilde{S} = \widehat{S}^{\textup{MLE}}$. 
    \item $L \geq L_0$. 
    Let
    \[
    a_0 = \frac{\varepsilon}{8n}, \qquad b = \frac{64L}{m}, \qquad J = \floor*{\log_2\rbra*{\frac{b}{2a_0}}}.
    \]
    For each $0 \leq j \leq J$, consider the Chebyshev polynomial estimator
    \[
    \widehat{S}^{\textup{Cheb}}_j = \sum_i f_{a_j, b, k_j, m}\rbra{N_i}, \qquad a_j = \frac{2^{j-3} \varepsilon}{n}, \qquad k_j = \ceil*{\sqrt{\frac{b}{a_j}}\log\rbra*{\frac{2}{\delta_j}}}, \qquad \delta_j = 2^{-j-2}. 
    \]
    Finally, set
    \[
    \widetilde{S} = \max\cbra*{ \widehat{S}^{\textup{MLE}}, \max_{0 \leq j \leq J} \frac{\widehat{S}^{\textup{Cheb}}_j - \delta_j n}{1 + \delta_j} }.
    \]
\end{enumerate}

The tester outputs ``Yes'' if $\widetilde{S} \leq n$ and ``No'' otherwise. 
Take $\widehat{S} = \min\cbra{\max\cbra{\ceil{\widetilde{S}}, 1}, n}$ as a pretty good lower bound on support size.

\section{Correctness}

\begin{lemma} \label{lemma:basic-properties}
    Suppose $L \geq L_0$.
    For every $0 \leq j \leq J$, 
    \begin{align}
        & \rbra{1-\delta_j} \min\cbra*{\frac{x}{a_j}, 1} \leq \E_{X \sim \operatorname{Poi}\rbra{mx}}\sbra*{ f_{a_j, b, k_j, m}\rbra{X} } \leq 1 + \delta_j, \qquad 0 \leq x \leq 1, \label{eq:deltaj} \\
        & \Abs*{ f_{a_j, b, k_j, m} }_\infty \leq 4 n^{1/16}. \label{eq:fnorm}
    \end{align}
    \cref{eq:deltaj} further gives 
    \begin{equation} \label{eq:ESj}
        \rbra{1-\delta_j} \sum_{i} \min\cbra*{\frac{p_i}{a_j}, 1} \leq \E\sbra*{ \widehat{S}^{\textup{Cheb}}_j } \leq \rbra{1 + \delta_j} \abs{\supp\rbra{p}}.
    \end{equation}
    If $\E\sbra{\widehat{S}^{\textup{MLE}}} \leq 2n$, then 
    \begin{equation} \label{eq:prob-union}
    \Pr\sbra*{ \abs*{\widehat{S}^{\textup{Cheb}}_j - \E\sbra{\widehat{S}^{\textup{Cheb}}_j} } \leq \delta_j n, ~ 0 \leq j \leq J } \geq 0.99 .
    \end{equation}
\end{lemma}
\begin{proof}
    It can be verified that the tuple $\rbra{a, b, k, m} \coloneqq \rbra{a_j, b, k_j, m}$ satisfies the conditions required in \cref{lemma:poly}: $m\sqrt{b\rbra{b-a_j}} \geq k_j$ and $a_j \leq b/2$. 
    \cref{eq:outcond2} in \cref{lemma:poly} gives 
    \[
    \rbra{1-\rho_j} \min\cbra{x/a_j, 1} \leq \E_{X \sim \operatorname{Poi}\rbra{mx}} \sbra*{f_{a,b,k,m}\rbra{X}} \leq 1 + \rho_j,
    \]
    which gives \cref{eq:deltaj} by noting that
    \[
    \rho_j = \rbra*{ \cosh\rbra*{2k_j \operatorname{artanh}\rbra*{\sqrt{\frac{a_j}{b}}}} }^{-1} \leq 2e^{-2k_j \sqrt{a_j/b}} \leq \frac{\delta_j^2}{2} \leq \delta_j.
    \]
    Denote $P_{a_j, b, k_j}\rbra{x} = \sum_{\ell=0}^{k_j} c_\ell x^\ell$.
    Using \cref{eq:outcond3} in \cref{lemma:poly}, \cref{eq:fnorm} is obtained by noting that for $1 \leq \ell \leq k_j$, 
    \[
    \abs*{f_{a_j, b, k_j, m} \rbra{\ell}} \leq 1 + \frac{\abs{c_\ell} \ell!}{m^{\ell}} \leq 1 + 15^{k_j} \rbra*{\frac{k_j}{mb}}^\ell \leq 1 + 15^{k_j} \leq 4n^{1/16},
    \]
    where we use that fact that $k_j < L/64$.

    Since $N_i$ are independent and $f_{a_j, b, k_j, m}\rbra{0} = 0$, we have
    \[
    \Var\sbra*{ \widehat{S}^{\textup{Cheb}}_j } \leq \Abs*{f_{a_j, b, k_j, m}}_\infty^2 \sum_i \Pr\sbra*{N_i > 0} \leq 16 n^{1/8} \E\sbra*{\widehat{S}^{\textup{MLE}}}.
    \]
    Note that this argument is also valid on a countable alphabet of $i$. 
    If $\E\sbra{\widehat{S}^{\textup{MLE}}} \leq 2n$, then, by Chebyshev's inequality, for any $0 \leq j \leq J$, 
    \[
    \Pr\sbra*{ \abs*{\widehat{S}^{\textup{Cheb}}_j - \E\sbra*{\widehat{S}^{\textup{Cheb}}_j}} > \delta_j n } \leq 32 n^{-7/8} \delta_j^{-2}.
    \]
    \cref{eq:prob-union} is obtained by the union bound and by noting that 
    \[
    \sum_{j=0}^J 32 n^{-7/8} \delta_j^{-2} < 64 n^{-7/8} \delta_J^{-2} \leq 64 n^{-7/8} \cdot \frac{L^4}{2^{32}} < 0.01,
    \]
    where we use $\delta_J^{-1} \leq L^2/2^{16}$. 
\end{proof}

\begin{lemma} \label{lemma:test-prob}
    If $\abs{\supp\rbra{p}} \leq n$, then 
    \begin{equation} \label{eq:prob1}
    \Pr\sbra*{ \widetilde{S} > \abs{\supp\rbra{p}} } < 0.01.
    \end{equation}
    For $1 \leq r \leq n$, if $\Delta_r\rbra{p} \geq \varepsilon$, then
    \begin{equation} \label{eq:prob2}
    \Pr\sbra*{ \widetilde{S} \leq r } < 0.01.
    \end{equation}
\end{lemma}
\begin{proof}
    If $L < L_0$, then $\widetilde{S} = \widehat{S}^{\textup{MLE}}$. 
    Then, $\widetilde{S} = \widehat{S}^{\textup{MLE}} \leq \abs{\supp\rbra{p}}$ holds with certainty, which gives \cref{eq:prob1}. 
    For $1 \leq r \leq n$, if $\Delta_r\rbra{p} \geq \varepsilon$, \cref{eq:prob-mle} in \cref{lemma:mle} gives \cref{eq:prob2}:
    \[
    \Pr\sbra*{ \widetilde{S} \leq r } = \Pr\sbra*{ \widehat{S}^{\textup{MLE}} \leq r } \leq \frac{r+1}{m\varepsilon} \leq \frac{2n}{m\varepsilon} \leq \frac{2L}{C_0} < 0.01. 
    \]

    In the rest, we assume that $L \geq L_0$. 
    To see \cref{eq:prob1}, since $\E\sbra{\widehat{S}^{\textup{MLE}}} \leq \abs{\supp\rbra{p}} \leq n \leq 2n$, \cref{eq:prob-union} in \cref{lemma:basic-properties} gives that it holds with probability at least $0.99$ that $\abs{\widehat{S}^{\textup{Cheb}}_j - \E\sbra{\widehat{S}^{\textup{Cheb}}_j} } \leq \delta_j n$ for all $0 \leq j \leq J$, in which case \cref{eq:ESj} in \cref{lemma:basic-properties} gives
    \[
    \frac{\widehat{S}^{\textup{Cheb}}_j - \delta_j n}{1 + \delta_j} \leq \frac{ \E\sbra{\widehat{S}^{\textup{Cheb}}_j} }{1 + \delta_j} \leq \abs{\supp\rbra{p}}.
    \]

    To show \cref{eq:prob2}, we consider two cases.
    \begin{enumerate}
        \item $\E\sbra{\widehat{S}^{\textup{MLE}}} > 2n$. 
        Since $\widetilde{S} \geq \widehat{S}^{\textup{MLE}}$, \cref{eq:basic} in \cref{lemma:mle}, together with Chebyshev's inequality, gives
        \[
        \Pr\sbra*{ \widetilde{S} \leq r } \leq \Pr\sbra*{ \widehat{S}^{\textup{MLE}} \leq n } \leq \frac{\Var\sbra{\widehat{S}^{\textup{MLE}}}}{\rbra{\E\sbra{\widehat{S}^{\textup{MLE}}} - n}^2} \leq \frac{\E\sbra{\widehat{S}^{\textup{MLE}}}}{\rbra{\E\sbra{\widehat{S}^{\textup{MLE}}} - n}^2} \leq \frac{2}{n} < 0.01.
        \]
        \item $\E\sbra{\widehat{S}^{\textup{MLE}}} \leq 2n$.
        Let $p_{\rbra{1}} \geq p_{\rbra{2}} \geq p_{\rbra{3}} \geq \dots$ be the decreasing probability weights of the distribution $p$. 
        We consider three cases on $p_{\rbra{r}}$.
        \begin{enumerate}
            \item[2.1.] $p_{\rbra{r}} \geq b/2$. Since $\widetilde{S} \geq \widehat{S}^{\textup{MLE}}$, the union bound gives
            \[
            \Pr\sbra*{ \widetilde{S} \leq r } \leq \Pr\sbra*{ \widehat{S}^{\textup{MLE}} \leq r } \leq \Pr\sbra*{ \exists 1 \leq i \leq r \colon N_{\rbra{i}} = 0 } + \Pr\sbra*{ \forall i > r \colon N_{\rbra{i}} = 0 },
            \]
            where $N_{\rbra{i}}$ is the count of the element with probability $p_{\rbra{i}}$.
            For the first term, 
            \[
            \Pr\sbra*{ \exists 1 \leq i \leq r \colon N_{\rbra{i}} = 0 } \leq \sum_{i=1}^r \Pr\sbra*{ N_{\rbra{i}} = 0 } = \sum_{i=1}^r e^{-mp_{\rbra{i}}} \leq r e^{-mp_{\rbra{r}}}. 
            \]
            For the second term, 
            \[
            \Pr\sbra*{ \forall i > r \colon N_{\rbra{i}} = 0 } = \Pr\sbra*{ \sum_{i > r} N_{\rbra{i}} = 0 } = e^{-m \Delta_r\rbra{p}},
            \]
            where note that $\sum_{i > r} N_{\rbra{i}} \sim \operatorname{Poi}\rbra{m \sum_{i > r} p_{\rbra{i}}} = \operatorname{Poi}\rbra{m \Delta_r\rbra{p}}$.
            Therefore, 
            \[
            \Pr\sbra*{ \widetilde{S} \leq r } \leq re^{-mp_{\rbra{r}}} + e^{-m\Delta_r\rbra{p}} \leq re^{-mb/2} + e^{-m\varepsilon} < 0.01.
            \]
            \item[2.2.] $p_{\rbra{r}} < a_0$. 
            By \cref{eq:ESj} in \cref{lemma:basic-properties}, 
            \[
            \E\sbra*{\widehat{S}^{\textup{Cheb}}_0} \geq \frac{1-\delta_0}{a_0} \sum_{i > r} p_{\rbra{i}} = \frac{1-\delta_0}{a_0} \Delta_r\rbra{p} \geq \frac{1-\delta_0}{a_0} \varepsilon = 6n.
            \]
            By \cref{eq:prob-union} in \cref{lemma:basic-properties}, it holds with probability at least $0.99$ that $\abs{\widehat{S}^{\textup{Cheb}}_0 - \E\sbra{\widehat{S}^{\textup{Cheb}}_0} } \leq \delta_0 n$, in which case 
            \[
            \widetilde{S} \geq \frac{\widehat{S}^{\textup{Cheb}}_0 - \delta_0 n}{1 + \delta_0} \geq \frac{ \E\sbra{\widehat{S}^{\textup{Cheb}}_0} - 2\delta_0 n}{1 + \delta_0} \geq \frac{6n - n/2}{1 + 1/4} = \frac{22}{5}n > r,
            \]
            which yields \cref{eq:prob2}.

            \item[2.3.] $a_0 \leq p_{\rbra{r}} < b/2$. 
            Let $0 \leq j \leq J$ such that $a_j \leq p_{\rbra{r}} < 2a_j$. 
            By \cref{eq:ESj} in \cref{lemma:basic-properties}, 
            \[
            \E\sbra*{\widehat{S}^{\textup{Cheb}}_j} \geq \rbra{1-\delta_j} \sum_{i} \min\cbra*{\frac{p_i}{a_j}, 1} \geq \rbra{1-\delta_j} \rbra*{r + \frac{\Delta_r\rbra{p}}{p_{\rbra{r}}}} \geq \rbra{1-\delta_j} \rbra*{r + \frac{\varepsilon}{2a_j}}.
            \]
            By \cref{eq:prob-union} in \cref{lemma:basic-properties}, it holds with probability at least $0.99$ that $\abs{\widehat{S}^{\textup{Cheb}}_j - \E\sbra{\widehat{S}^{\textup{Cheb}}_j} } \leq \delta_j n$, in which case 
            \begin{align*}
            \widetilde{S} 
            & \geq \frac{\widehat{S}^{\textup{Cheb}}_j - \delta_j n}{1 + \delta_j} \geq \frac{ \E\sbra{\widehat{S}^{\textup{Cheb}}_j} - 2\delta_j n}{1 + \delta_j} \geq \frac{\rbra{1-\delta_j} \rbra{r + \varepsilon/2a_j} - 2\delta_jn}{1 + \delta_j} \\
            & = r + \frac{\delta_j \rbra{\rbra{14-16\delta_j}n - 2r}}{1+\delta_j} > r,
            \end{align*}
            where we use $\delta_j \leq 1/4$ and $r \leq n$. 
            This yields \cref{eq:prob2}.
        \end{enumerate}
    \end{enumerate}
\end{proof}

\begin{proof}[Proof of \cref{thm:support-size-testing}]

When $\abs{\supp\rbra{p}} \leq n$, \cref{lemma:test-prob} gives that $\Pr\sbra{\widetilde{S} > n} \leq \Pr\sbra{\widetilde{S} > \abs{\supp\rbra{p}}} < 0.01$.
When $\Delta_n\rbra{p} \geq \varepsilon$, \cref{lemma:test-prob} gives $\Pr\sbra{\widetilde{S} \leq n} < 0.01$. 
\end{proof}

\begin{proof}[Proof of \cref{thm:pretty}]
When $\abs{\supp\rbra{p}} \geq n$, $\widehat{S} \leq n \leq \abs{\supp\rbra{p}}$ holds with certainty. 
When $\abs{\supp\rbra{p}} < n$, \cref{lemma:test-prob} gives $\Pr\sbra{\widehat{S}>\abs{\supp\rbra{p}}}\leq0.01$.
Therefore, $\Pr\sbra{\widehat{S} \leq \abs{\supp\rbra{p}}} \geq 0.99$. 

For the left hand side of the inequality, let $r = \min\cbra{\operatorname{eff}_\varepsilon\rbra{p}, n} - 1$. 
If $r = 0$, then $\widehat{S} \geq 1 = r + 1$ with certainty. 
If $1 \leq r \leq n - 1$, then $\Delta_r\rbra{p} > \varepsilon$. 
\cref{lemma:test-prob} gives $\Pr\sbra{ \widetilde{S} \leq r } < 0.01$, which is $\Pr\sbra{\widehat{S} \geq \min\cbra{\operatorname{eff}_\varepsilon\rbra{p}, n}} \geq 0.99$. 

Combining the above with the union bound gives $\Pr\sbra{\min\cbra{\operatorname{eff}_\varepsilon\rbra{p}, n} \leq \widehat{S} \leq \abs{\supp\rbra{p}}} \geq 0.98$. 
\end{proof}

\section*{Acknowledgment}

The author used Large Language Models as AI-assisted research and writing tools throughout the
preparation of this paper. The tools were used to help brainstorm ideas and explore proof
strategies. 
All final mathematical claims, algorithms, proofs, citations, and wording were written by the author. 
The author assumes responsibility for all content of the paper.

The author thanks Jinge Bao, Minbo Gao, and Penghui Yao for communication regarding the related work \cite{BGY26}. 

\addcontentsline{toc}{section}{References}

\bibliographystyle{alphaurl}
\bibliography{main}

\appendix

\section{Proof of Lemma \ref{lemma:poly}} \label{app:proof-lemma-poly}

\begin{proof}
    By definition, direct calculation gives \cref{eq:def-Q}:
    \begin{align*}
        Q\rbra{x} 
        & = \sum_{\ell = 0}^\infty \frac{e^{-mx} \rbra{mx}^\ell}{\ell!} f\rbra{\ell} \\
        & = \sum_{\ell = 1}^k \frac{e^{-mx} \rbra{mx}^\ell}{\ell!} \rbra*{1 + \frac{c_\ell \ell!}{m^\ell} } + \sum_{\ell = k+1}^\infty \frac{e^{-mx} \rbra{mx}^\ell}{\ell!} \\
        & = \sum_{\ell = 1}^\infty \frac{e^{-mx} \rbra{mx}^\ell}{\ell!} + \sum_{\ell = 1}^k \frac{e^{-mx} \rbra{mx}^\ell}{\ell!} \cdot \frac{c_\ell \ell!}{m^\ell} \\
        & = \rbra{1 - e^{-mx}} + e^{-mx} \rbra{P_{a,b,k}\rbra{x} + 1} \\
        & = 1 + e^{-mx} P_{a,b,k}\rbra{x}.
    \end{align*}

    Now we prove \cref{eq:outcond2}. 
    We consider three cases.
    \begin{enumerate}
        \item $0 \leq x \leq a$. 
        Since $T_k\rbra{x}$ is increasing and convex on $[1, +\infty)$, $P_{a,b,k}\rbra{x}$ is increasing and concave on $\sbra{0, a}$, which gives
    \[
    P_{a,b,k}\rbra{x} \geq \rbra*{1 - \frac{x}{a}} P_{a,b,k}\rbra{0} + \rbra*{\frac{x}{a}} P_{a,b,k}\rbra{a} = -1 + \frac{\rbra{1-\rho}x}{a}, \qquad 0 \leq x \leq a,
    \]
    where we use $P_{a,b,k}\rbra{0} = -1$ and $P_{a,b,k}\rbra{a} = -\rho$.
    Moreover, since $P_{a,b,k}\rbra{x} \leq 0$ and $e^{-mx} \leq 1$ for $0 \leq x \leq a$, \cref{eq:def-Q} gives
    \begin{equation} \label{eq:case1}
    \frac{\rbra{1-\rho}x}{a} \leq Q\rbra{x} \leq 1, \qquad 0 \leq x \leq a.
    \end{equation}
        \item $a \leq x \leq b$.
        Since $\abs{P_{a,b,k}\rbra{x}} \leq \rho$, \cref{eq:def-Q} gives
    \begin{equation} \label{eq:case2}
    1 - \rho \leq Q\rbra{x} \leq 1 + \rho.
    \end{equation}
        \item $x \geq b$. 
        Let $u = \frac{2\rbra{x-b}}{b-a} \geq 0$.
        Note that $\frac{b+a-2x}{b-a} = -1-u$. 
    Then, 
    \[
    T_k\rbra{1+u} = \cosh\rbra*{k \operatorname{arcosh}\rbra{1+u}} \leq e^{k \operatorname{arcosh}\rbra{1+u}} \leq e^{k\sqrt{2u}}. 
    \]
    On the other hand, the condition $m\sqrt{b\rbra{b-a}} \geq k$, together with the fact that $x = b + \frac{b-a}{2} u \geq \sqrt{2b\rbra{b-a}u}$, implies $mx \geq k\sqrt{2u}$.
    Therefore,
    \begin{equation} \label{eq:case3}
    \abs{Q\rbra{x} - 1} = \rho e^{-mx} T_k\rbra{1+u} \leq \rho e^{k\sqrt{2u}-mx} \leq \rho.
    \end{equation}
    \end{enumerate}
    Combining \cref{eq:case1,eq:case2,eq:case3} gives \cref{eq:outcond2}.

    Finally, we prove \cref{eq:outcond3}.
    Let $u = x/b$. 
    Then, 
    \[
    P\rbra{x} = P\rbra{bu} = -\rho T_k\rbra{\alpha - \beta u}, \qquad \alpha = \frac{a+b}{b-a} \leq 3, \qquad \beta = \frac{2b}{b-a} \leq 4.
    \]
    Let $M_j = \Abs{T_j\rbra{\alpha - \beta u}}_1$ be the absolute sum of the coefficients of $T_j\rbra{\alpha - \beta u}$ in terms of $u$. 
    Then, $M_0 = 1$, $M_1 \leq 7$, and $M_{j+1} \leq 14M_j + M_{j-1}$, which gives $M_j \leq 15^j$ for any integer $j \geq 0$. 
    Therefore, 
    \[
    \abs{c_\ell} b^\ell \leq \sum_{\ell = 0}^k \abs{c_\ell} b^\ell = \Abs{P\rbra{bu}}_1 = \rho M_k \leq 15^k,
    \]
    which gives \cref{eq:outcond3} using $\ell! \leq k^\ell$ for $0 \leq \ell \leq k$. 
\end{proof}

\end{document}